\documentclass[a4paper,10pt, leqno]{article}
\usepackage{amsfonts, amsmath, amsthm}
\usepackage{cite}
\usepackage{graphicx}
\usepackage{setspace}
\newtheorem{theorem}{Theorem}[section]
\newtheorem{definition}[theorem]{Definition}
\newtheorem{lemma}[theorem]{Lemma}
\newtheorem{example}[theorem]{Example}

\newtheorem{remark}[theorem]{Remark}
\usepackage{enumerate}
\title{Price-Setting of Market Makers: A Filtering Problem with an Endogenous Filtration}
\author{Christoph K\"uhn\footnote{Frankfurt MathFinance Institute, Goethe-Universit\"at, D-60054 Frankfurt a.M., Germany,
e-mail: \{ckuehn, riedel\}{\char'100}math.uni-frankfurt.de\newline
The second named author acknowledges the financial support by the Deutsche Forschungsgemeinschaft~(DFG).} 
\quad and\quad  Matthias Riedel\footnotemark[\value{footnote}]\\ }
\date{}
\numberwithin{equation}{section}

\newcommand{\bbr}{\mathbb{R}}
\newcommand{\bbn}{\mathbb{N}}

\newcommand{\bbq}{\mathbb{Q}}
\newcommand{\bbf}{\mathbb{F}}

\newcommand{\eps}{\epsilon}
\newcommand{\lam}{\lambda}
\newcommand{\rate}{\lambda}

\newcommand{\F}{\mathcal{F}}
\newcommand{\N}{\mathcal{N}}

\newcommand{\B}{\mathcal{B}}
\newcommand{\A}{\mathcal{A}}

\newcommand{\ask}{\overline{S}}
\newcommand{\bid}{\underline{S}}
\newcommand{\cdl}{c\`adl\`ag }
\newcommand{\norm}[1]{\left\|#1\right\|}

\begin{document}
\maketitle

\begin{abstract}
We study the price-setting problem of market makers under risk neutrality and perfect competition in continuous time. Thereby we follow the classic Glosten-Milgrom model~\cite{glosten1985bid} 
that defines bid and ask prices as expectations of a true value of the asset given the market makers' partial information that includes the customers trading decisions. 
The true value is modeled as a Markov process that can be observed by the customers with some noise at Poisson times. 

We analyze the price-setting problem in a mathematically rigorous way by solving a filtering problem with an endogenous filtration that depends on the bid and 
ask price process quoted by the market maker. Under some conditions we show existence and uniqueness of the price processes.
\end{abstract}

\begin{tabbing}
\end{tabbing}

\section{Introduction}
In specialist markets one or several market makers (also called specialists) provide liquidity by offering to buy or to sell the respective asset at any time. 
They quote both a bid price at which they commit themselves to buy 
and a higher ask price at which they sell. 
By doing so market makers face certain risks for which they are compensated by the bid-ask spread. 

The risk can be decomposed mainly into two components: inventory and information risk. Inventory risk describes the risk that market makers or other liquidity providers 
accumulate large positive or 
negative inventories in the respective asset and then prices move against them. In a continuous time framework this was studied by Ho and Stoll~\cite{ho1981optimal}
and, more recently, further developed as optimal stochastic control problems by Avellaneda and Stoikov~\cite{avellaneda2008high}, 
Guilbaud and Pham~\cite{guilbaud2012optimal}, Veraart~\cite{veraart2010optimal} and Cartea and Jaimungal~\cite{cartea2012risk} among others. 

The other risk market makers take is information risk, i.e. the risk that at least part of the customers have superior (or insider) information about 
the hidden true value of the asset and trade strategically to their own advantage and therefore to the disadvantage of the market maker. Thus, the market maker faces 
an adverse selection problem. Although the nature of the two types of risk is quite different, their effects are somehow similar. 
Namely, if a customer buys assets, the market maker will most likely raise both 
his bid and his ask price. On the one hand, because he wants to avoid further buying and stimulate the sell-side to control his inventory, and on the other hand because 
he believes that the purchase of the 
customer has conveyed some good news about the true value of the asset. 
Here, assuming that market makers are risk neutral, we concentrate on information risk, 
which was first studied by Copeland and Galai~\cite{copeland1983information} and more general and in continuous time by Glosten and Milgrom~\cite{glosten1985bid} who describe 
the prices as expectations of a hidden true value. This zero expected profit condition can be explained by risk neutrality and perfect competition among market makers. 
It leads to quite tractable models and may still be used as a benchmark for more involved situations.
An alternative approach is by  Kyle~\cite{kyle1985continuous} (developed further by Back~\cite{back1992insider})
who not only modelled
how the market makers handle the information flow from customers, but who also considers a strategically behaving insider optimally using his knowledge to his own advantage. 
However, in contrast to  Glosten-Milgrom, Kyle models a single price process and can therefore not explain the bid-ask spread. 
A connection to the Glosten-Milgrom model 
was established by Krishnan~\cite{krishnan1992equivalence} and more general by Back and Baruch~\cite{back2004information}.

Already in a static model showing or disproving the existence or the 
uniqueness of Glosten-Milgrom prices is a non-trivial issue and there are only quite few substantial contributions.
Bagnoli, Viswanathan, and Holden~\cite{BVH2001} derive necessary and sufficient conditions for the existence of a so-called linear equilibrium  
in a one-period model with several strategically behaving insiders. Linearity means that,
after observing the size of the arriving market order, the market maker quotes a price per share which is affine linear, but not constant, in the order size. 
The market maker can draw conclusions from the order size
about the typ of trader submitting the order. It turns out that linear equilibria only exist in special cases. 
Back and Baruch~\cite{back2004information} derive (in)equalities under which they prove the existence of an equilibrium in the continuous time Gloston-Milgrom model with a 
strategically behaving insider and
two possible states of the true asset value. Then, it is shown numerically that the (in)equalities have a solution and an equilibrium is constructed. 
The decision making in our model is very similar to Das~\cite{das2005learning, das2008effects}, who provides methods to simulate the Glosten-Milgrom price process in 
a discrete time model and examines some statistical properties of the prices in the market model numerically. 

We develop a mathematically rigorous continuous time Glosten-Milgrom model by solving a filter problem with an endogenous filtration
and show existence and uniqueness of the price processes under some conditions.
The bid and ask prices of the market maker are determined by the zero profit condition given his information about the time-dependent true value of the asset. 
However, this information, i.e. the filtration, depends again on the prices he sets, thus, he influences the learning environment by setting bid and ask prices and there appears 
a fixed point problem. If, for example, the market maker sets a very large spread, there will be only a small amount of trades on which he can base his estimation of the true value.  
Mathematically this means that the filter problem is w.r.t. a filtration that is not exogenously given but that is part of the solution. The filtration depends on the bid 
and ask price process which have for their part to be predictable w.r.t. the filtration. 
This is an essential difference to other filter problems in market microstructure models with a not directly observable true value of the asset where, however, also point processes are 
used, see e.g. the article by Zeng~\cite{zeng2003partially}.
We show that Glosten-Milgrom bid and ask price processes are  
fixed points of certain functionals acting on the set of stochastic processes and they are given by some deterministic functions of the conditional probabilities of the true value 
process (under the resulting partial information of the market maker). The conditional probabilities can be obtained as the solution of a system of SDEs.

In the literature on market making filtering problems with an endogenous filtration already appear in many articles on the Kyle model (and its generalizations), see  
Back~\cite{back1992insider}, Back and Baruch~\cite{back2004information}, Lasserre \cite{lasserre.2004}, Aase, Bjuland, and {\O}ksendal~\cite{aase.bjuland.oksendal.2012},
and Biagini, Hu, Meyer-Brandis, and {\O}ksendal~\cite{BHMO2012}, among others. In the Kyle model, a rational price process is characterized as the conditional expectation 
of the true value of the asset under the filtration 
of the market maker which itself depends on 
the price process through the demand of the insider. But, the inherent fixed point problem which is solved in a Brownian setting is fundamentally different (to the problem we solve)
as accumulated purchases and sells are continuous processes and new information arise continuously. In addition, note that the Kyle model cannot explain the bid-ask spread as
it models a single price process at which both buy and sell orders are executed.

The paper is organized as follows. In Section~\ref{model} the continuous time model is introduced and the main result (Theorem~\ref{main}) is stated.
Section~\ref{static} considers the static case. Under certain conditions we prove an existence and uniqueness result (Theorem \ref{solstat}).
In Section~4, we prove Theorem~\ref{main} using the results in Section~\ref{static}.

\section{The model and the main result} \label{model}

In the following we will develop a general model in continuous time for a specialist market, i.e. a market where a market maker or specialist offers to buy or sell at any point in time 
to the bid and ask prices he quotes.

All random variables that we introduce live on the probability space~$(\Omega, \F, P)$ whereas different filtrations are considered. We assume that the \cdl process
$X=(X_t)_{t\geq 0}$, interpreted as the time-dependent true value of the asset, is a time-homogeneous Markov process with finite state space $\{x_1,\ldots,x_n\}$, $n\geq2$ where $x_{\min}=x_1 < \ldots < x_n=x_{\max}$,
and has transition kernel
\[
q(i,j):=\lim_{t\to0} \frac{1}{t} P[X_t=j\ |\ X_0 = i].
\]
The market maker knows the distribution of $X$ but does not know the actual value. The only source of information
which is available to the market maker are the trades that take place at the prices he sets.

To model the customer flow, let $N$ be a Poisson process with rate $\lam>0$. We denote the ordered jump times of $N$ by $\tau_1< \tau_2<\tau_3,\ldots$.
We assume that at these times potential customers arrive at the market (unseen by the market maker). The customers have some disturbed information about the true value of the asset 
which is given by $X_{\tau_i} + \eps_i$ for the $i$-th customer where $(\eps_i)_{i\in \bbn}$ is a sequence of i.i.d. random variables. We assume that $X$,\ $N$, 
and $(\eps_i)_{i\in \bbn}$ are independent of each other.

We further assume that the market maker sets a pair of prices according to an $\F \otimes \B ([0, \infty))$-measurable mapping $S: \Omega \times [0,\infty) \to \bbr^2$. 
We write $S=(\ask,\bid)$ to denote ask and bid prices and we only admit prices with $\ask_t(\omega)\ge \bid_t(\omega)$ for all $(\omega,t)$. To be economically meaningful the strategy $S$ has to satisfy some predictability condition that will be given in Definition \ref{defad}.

A potential customer buys one asset if $X_{\tau_i}+\eps_i \geq \ask_{\tau_i}$ and sells one asset if $X_{\tau_i}+\eps_i \leq \bid_{\tau_i}$. He does nothing if his valuation is within the spread. 

In the decision making of the customers we follow Das \cite{das2005learning}. In the original Glosten-Milgrom paper \cite{glosten1985bid} a buy, say, occurs if 
$\rho_t E[X|\A] \geq \ask_{t}$, where $\rho_t$ is an independent random variable which represents time-preference and plays the role of $\eps_i$ in our model. 
$\rho_t \gg 1$ means that an impatient buyer arrives and $\rho_t \ll 1$ stands for an impatient seller.  
The sigma-algebra $\A$ represents the partial information of the insider. For $\A= \sigma(X)$, the models, including possible interpretation of $\eps_i$ and $\rho_t$, are quite similar. 
Further, note that the behavior of the customers is not rational. A rational exploitation of the given information would be to buy if $E[X_{\tau_i}|X_{\tau_i}+\eps_i] \geq \ask_{\tau_i}$. 
A high realization of $X_{\tau_i}+\eps_i$ might simply mean that $\eps_i$ is large, which the costumer may be well aware of if he knows the distributions of $X_{\tau_i}$ 
and $\eps_i$ separately. 
It was shown by Milgrom and Stokey \cite{milgrom1982information} that there has to be some irrational behavior for a price to exist. Very often in information-based models 
(for example in the famous Kyle model \cite{kyle1985continuous}) this irrational behavior is introduced by the assumption that there are two types of traders: Those who trade 
on superior information called insiders (with $\eps =0$) and those who trade for liquidity reasons, sometimes called noise traders (with $\eps=\pm \infty$). This describes a limiting case of the model we consider here, where customers have all kinds of noise or preference in their valuation.

Note that the volume of each trade is set to one. Hence, we ignore any volume effect. It is an disputed question among economists whether the volume of a trade has some information content (cf. \cite{o2007market}, p.160 ff.).

Let $B_0=C_0=0$. We introduce the sequence of random times of actual buys by 

\[B_i:=\inf\{\tau_j|  \tau_j > B_{i-1} , X_{\tau_j} + \epsilon_j \geq \ask_{\tau_j}\}, \qquad i\geq 1, \]

and a sequence of actual sells by 

\[C_i:=\inf\{\tau_j|  \tau_j > C_{i-1} , X_{\tau_j} + \epsilon_j \leq \bid_{\tau_j}\}, \qquad i\geq 1. \]

In addition we define the counting processes of actual buys and sells by

\begin{equation}\label{defN}
  N^B_t:= \sum_{i\geq 1} 1_{\{B_i \leq t\}} \qquad \mbox{and} \qquad N^C_t:= \sum_{i\geq 1} 1_{\{C_i \leq t\}}.
\end{equation}

The filtration of the market maker is given by $\bbf^S =(\F_t^S)_{t\geq 0}$, where
\begin{equation}\label{deffs}
\F_t^{S} := \sigma(\{B_i \leq s\}, \{C_i \leq s\}, s\leq t, i \in \mathbb{N})= \sigma\left(N^B_s, N^C_s, s\leq t\right).
\end{equation}

Since $\bbf^S$ is generated by counting processes, it is a right-continuous filtration  (see Theorem I.25 in \cite{protter2004stochastic}). However, it does not in general satisfy the usual conditions, since the null sets are not necessarily included. 

 \begin{figure}[htbp]
\begin{center}
 \includegraphics[scale=0.65]{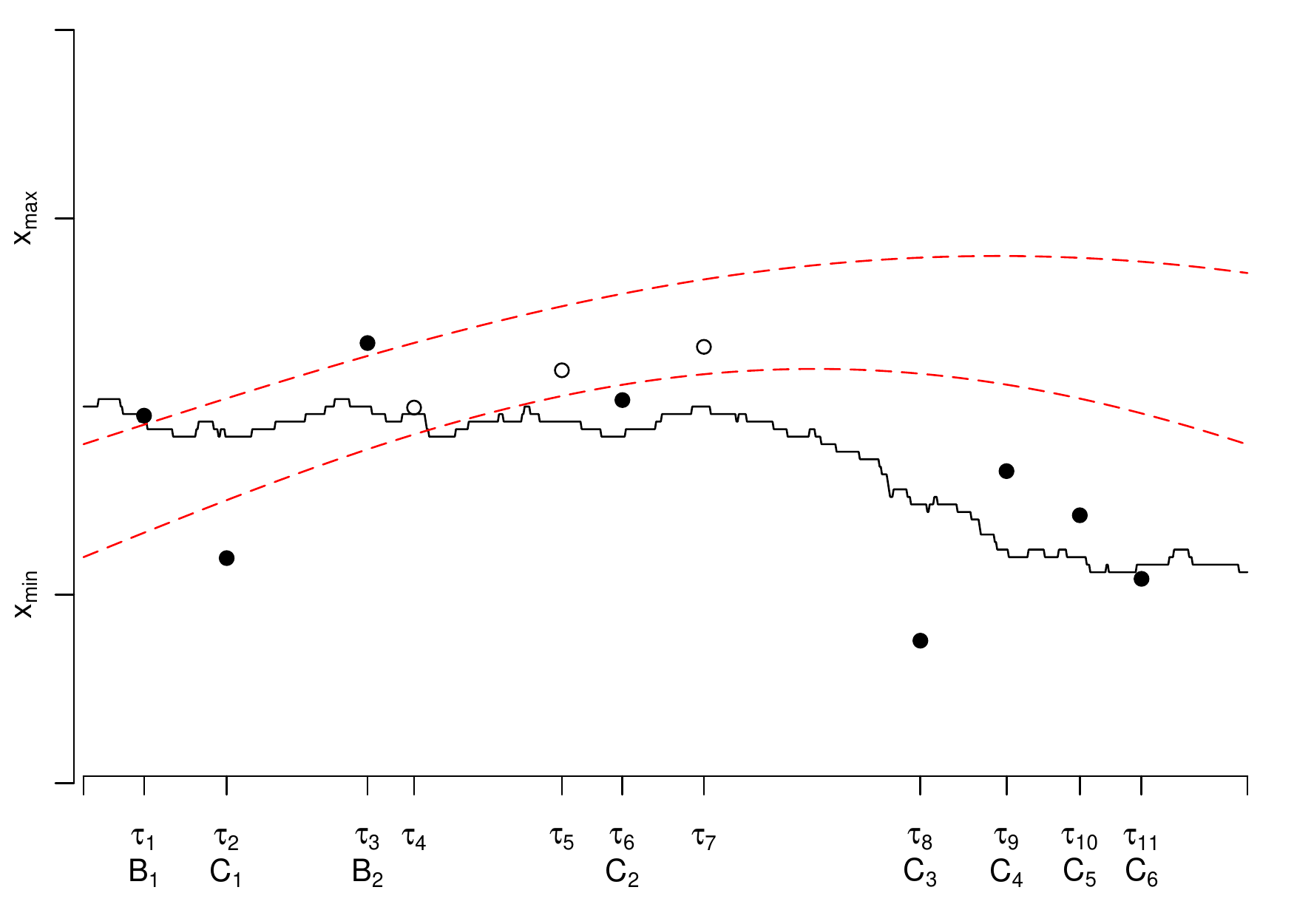}
\end{center}
 \caption{The black line represents the true value $X$ and some quoted prices $\ask \geq \bid$ (here not at equilibrium) are given by the dotted red lines. All potential trades $X_{\tau_i}+\eps_i$ are given by the bullets, which are filled if a trade takes place at $\ask$ or $\bid$.}
 \label{fig1}
 \end{figure}

>From an economic viewpoint pricing strategies of market makers make sense only if they are $\bbf^S$-predictable, as $\bbf^S$ is the information flow of the market maker.

\begin{definition}\label{defad}
We say that $S$ is an admissible pricing strategy if it is $\bbf^S$-predictable and $x_{\max} \geq \ask_t(\omega) \geq \bid_t(\omega)\geq x_{\min}$ for all $(\omega,t) \in \Omega \times \bbr_+$.
\end{definition}

We impose the restriction that the prices lie between $x_{\min}$ and $x_{\max}$ because otherwise there would be either arbitrage opportunities or no trades at all.
Note that the definition is quite implicit, since the filtration $\bbf^S$ depends itself on $S$.

The model stated above gives a natural, though complex, framework to examine price-setting of market makers. We now proceed to consider a certain type of price-setting which involves the 
Glosten-Milgrom idea of risk neutrality and
perfect competition between market makers.

\begin{definition}\label{admis}
We say that an admissible pricing strategy $S$ is a Glosten-Milgrom pricing strategy (GMPS) if 

\begin{equation}\label{20.9.2012.1}
E\left[\sum_{B_i \leq \tau} (\ask_{B_i}-X_{B_i})\right]= 0 \mbox{ and }  E\left[\sum_{C_i \leq \tau} (\bid_{C_i}-X_{C_i})\right]=0 
\end{equation}
for every bounded $\bbf^S$-stopping time $\tau$.
\end{definition}

Each summand in (\ref{20.9.2012.1})  is bounded by $x_{\max}-x_{\min}$ and the sequence of $B_i$ and $C_i$ is included in the Poisson times. This yields integrability of the sums. 
Note that this definition implies that not only the whole business makes zero profits but both the buy-side and sell-side business separately. We assume that 
in no stochastic time interval it is possible to make a gain in expectation. By perfect competition among market makers it is not possible to offset a loss to obtain overall zero profits.

\begin{theorem}\label{main}
Let $C:=x_{\max}-x_{\min}$ and $\Phi(y):=P[\eps_1 \geq y]$,\ $y\in\bbr$.
Assume that $\Phi$ is differentiable (i.e. the distribution of $\eps_1$ has density $-\Phi'$) on $[-C,C]$, $1>\Phi(0)>0$,  and
\[-\Phi'(y)\leq \frac{K}{C} \min\{\Phi(y), 1-\Phi(y)\}\]
for all $y\in [-C,C]$ and a constant $K< 1$. Then, there exists a Glosten-Milgrom pricing strategy and it is unique up to a  $(P\otimes\lam)$-null set, 
where $\lam$ denotes the Lebesgue measure on $\bbr_+$.
\end{theorem}
The theorem is proven in Section \ref{proof}. 

\section{Glosten-Milgrom prices in a static model}\label{static}
As a first step to prove Theorem \ref{main} we will consider a static version of the dynamic Glosten-Milgrom model introduced in Section \ref{model} that also illustrates the idea of Glosten-Milgrom prices. It examines the situation at a time when a potential customer arrives at the market in the continuous time model.

In this section, let $X$ be a real-valued random variable which represents the true value of some asset. We assume that $X$ is unknown to all market participants, but the customer has a disturbed 
valuation given by $X + \eps$, where $\eps$ represents some observation error or time preference and is independent of $X$.  For the rest of the section we will only consider ask prices, 
since bid and ask prices can be determined independently from each other and bid prices are developed completely analogous. The independency of the price-setting problems is in 
contrast to the dynamic case, where some interdependency occurs as the filtration contains the information of both, buys and sells. We assume that a potential customer buys 
if his valuation is higher than the ask-price $s$. Thus, the profit of the market maker is given by $(s-X)1_{\{X+\eps \geq s\}}$. Again, 
we assume that the price-setting must confine a  zero-expected-profit condition. 

\begin{definition} We say that $s\in \bbr$ is a static Glosten-Milgrom ask price if
 \begin{equation} \label{fixstat}
 E[(s-X)1_{\{X+\eps \geq s\}}] = 0.
\end{equation}
\end{definition}

The question is now whether solutions to (\ref{fixstat}) exist and if so, whether they are unique. Roughly speaking, a Glosten-Milgrom-price exists and is unique if the tails of $\eps$ are ``heavy enough`` in comparison to those of $X$. Let us start with two simple examples that show cases of non-existence and non-uniqueness.

\begin{example}
Let $\eps =0$ and assume that $X$ is not essentially bounded from above. Then there exist no $s\in\bbr_+$ with $E[(s-X)1_{\{X \geq s\}}] = 0$, since the integrand is always 
non-positive and negative with positive probability. For $\eps =0$ and $X$ essentially bounded by $x_{\max}$ only $s=x_{\max}$ is a (trivial) solution in $(-\infty,x_{\max}]$.
\end{example}

\begin{example}
 Let $\eps$ be $1$ and $-1$ with probability $\frac{1}{2}$ and $X$ be $1$ with probability $\frac{3}{4}$ and $3$ with probability $\frac{1}{4}$, then $\frac{9}{5}$ and $3$ are both solutions of (\ref{fixstat}) and hence, the static Glosten-Milgrom ask price is not unique.
\end{example}

\begin{lemma}\label{allwbuy}
Let $X$ be bounded between $x_{\min}$ and $x_{\max}$ a.s. and define $C:=x_{\max}-x_{\min}$. Let $\Phi(y):= P[\eps \geq y]$ be the inverse distribution function of $\eps$. If $\Phi$ is differentiable (i.e. the distribution of $\eps$ has density $-\Phi'$) on $[-C,C]$, $\Phi(0)>0$  and 
\begin{equation} \label{condstat}
-\Phi'(y)\leq \frac{K}{C} \Phi(y)
\end{equation} for all $y\in [-C,C]$ and a constant $K< 1$, it follows that $\Phi(C)>0$, which implies that $P[X+\eps \geq s]>0$, i.e. the probability that a buy occurs is strictly larger than $0$, for all prices $s \leq x_{\max}$. 
\end{lemma}

\begin{proof}
Remember that $\Phi$ is in $[0,1]$ and decreasing. We have
\begin{align*}
\Phi(C)  &= \int_0^C \Phi'(t)dt +\Phi(0)  \geq  \frac{K}{C} \int_0^C - \Phi(t) dt +\Phi(0) \\
& \geq -\frac{K}{C} C  \Phi(0) +\Phi(0) \geq (1-K) \Phi(0)  > 0,
\end{align*}
since $K<1$ and $\Phi(0)>0$. Furthermore we have for all $s\le x_{\max}$
\[P[X+\eps \geq s] = P[\eps \geq s-X] \geq P[\eps \geq x_{\max}-x_{\min}] =  \Phi(C)>0. \qedhere \]
\end{proof}

Under the assumptions of Lemma \ref{allwbuy} we make the following definition.

\begin{definition} \label{defgh}
We indicate the distribution of $X$ by $\pi$. For $s \in [x_{\min}, x_{\max}]$ we define
\[g(s,\pi):= E[X|X+\eps \geq s]:= \frac{E\left[X 1_{\{X+\eps \geq s\}}\right]}{P\left[X+\eps \geq s\right]}=\frac{E[X \Phi(s-X)]}{E[ \Phi(s-X)]} .\]
\end{definition}

Lemma \ref{allwbuy} ensures that $g$ is well-defined for every $\pi$. Now, the zero-profit condition (\ref{fixstat}) translates to 

\[g(s,\pi)= s.\]

Thus, for given $\pi$, the question of existence and uniqueness of a Glosten-Milgrom ask price is the same as the existence and uniqueness of a fixed point of $g$, which leads us the way to prove the following theorem.

\begin{theorem} \label{solstat}
Let all assumptions of Lemma \ref{allwbuy}  be fulfilled. Then there exists an unique static Glosten-Milgrom ask price in $[x_{\min},x_{\max}]$.
\end{theorem}

For parametric families of distributions of $\eps$ (as e.g. the normal distribution) the central condition (\ref{condstat}) can usually be secured by choosing parameters such that the variance is high. In economic terms this corresponds to customers whose information is less precise or who are impatient. We still allow that $\eps$ takes the values $\pm \infty$. 

Note that we make no assumptions on the distribution of $X$ apart from the boundedness but quite explicit assumptions on the distribution of $\eps$. The fact that we have existence and uniqueness for all distributions of $X$ with compact support will be central in the continuous time model.

\begin{proof} 
Since $\pi$ is fixed we omit it. We consider the derivative of 
\[g(s)= E[X|X+\eps \geq s] = \frac{E[X \Phi(s-X)]}{E[ \Phi(s-X)]}\]

for $ x_{\min}\leq s\leq x_{\max}$ which is given by

\begin{align*}
g'(s)&= \frac{E_X[X \Phi'(s-X)]E_Z[\Phi(s-Z)]- E_Z[Z \Phi(s-Z)]E_X[\Phi'(s-X)]}{(E_X[\Phi(s-X)])^2}\\
&=\frac{E_X[E_Z [X \Phi'(s-X)\Phi(s-Z)- Z \Phi(s-Z)\Phi'(s-X)]]}{E_X[E_Z[\Phi(s-X)\Phi(s-Z)]]}\\
&=\frac{E_X[E_Z [-\Phi'(s-X)\Phi(s-Z)(Z- X)]]}{E_X[E_Z[\Phi(s-X)\Phi(s-Z)]]}\\
&\leq \frac{E_X[E_Z [-\Phi'(s-X)\Phi(s-Z)|Z- X|]]}{E_X[E_Z[\Phi(s-X)\Phi(s-Z)]]}\\
&\leq  C \frac{E_X[E_Z [-\Phi'(s-X)\Phi(s-Z)]]}{E_X[E_Z[\Phi(s-X)\Phi(s-Z)]]}\\
&\leq  C \frac{K}{C}\frac{E_X[E_Z [ \Phi(s-X)\Phi(s-Z)]]}{E_X[E_Z[\Phi(s-X)\Phi(s-Z)]]}\\
& = K,
\end{align*}
for $K$ from (\ref{condstat}) where $Z$ is an independent copy of $X$. Hence, $0\leq g'(s) \leq K<1$ for all $s\in [x_{\min}, x_{\max}]$  and therefore 

\begin{equation}\label{contra}
 |g(s)-g(t)| \leq K |s-t|.
\end{equation}

This means that $g$ is a contraction which has a unique fixed point by the Banach fixed point theorem.
\end{proof} 

We have already seen that $g$ is Lipschitz-continuous in $s$ with parameter $K<1$ in (\ref{contra}). If $X$ has discrete distribution we further obtain Lipschitz-continuity 
in the distribution $\pi$ (which we will use in the proof of  Theorem~\ref{main}).

\begin{lemma}\label{lippi}
Let all assumptions of Lemma \ref{allwbuy} be fulfilled. In addition assume that $X$ takes only finitely many values, i.e. there exist $x_{\min}=x_1 < ... < x_n=x_{\max}$ and $\pi = (\pi_1, ..., \pi_n)$ such that $P[X=x_i] = \pi_i$ for all $i$ and $\sum_{i=1}^n \pi_i =1$.  Then 

\[|g(s,\pi)-g(\widetilde{s},\widetilde{\pi})|\leq K |s-\widetilde{s}| + L \sum_{i=1}^n|\pi_i-\widetilde{\pi}_i|\]
for $K$ from (\ref{condstat}) and $L=\frac{2 x_{\max}}{\Phi(C)^2}<\infty$, all $s,\widetilde{s}\in [x_{\min},x_{\max}]$ and all distributions $\pi,\widetilde{\pi}$. 
\end{lemma}

\begin{proof} 
First, we see that
\begin{equation}\label{fineq}
 \begin{aligned}
 |g(s,\pi)-g(\widetilde{s},\widetilde{\pi})| & = |g(s,\pi)- g(s,\widetilde{\pi}) + g(s,\widetilde{\pi}) -g(\widetilde{s},\widetilde{\pi})| \\
& \leq |g(s,\pi) -g(s,\widetilde{\pi})| + K |s- \widetilde{s}| 
\end{aligned}
\end{equation}
 
by (\ref{contra}). It remains to show that 
\[|g(s,\pi)-g(s,\widetilde{\pi})| \leq  L \sum_{i=1}^n|\pi_i-\widetilde{\pi}_i|.\]

To shorten notation we write
\[\alpha(f(X),\pi):= E[ f(X) \Phi(s-X)] = \sum_{i=1}^n \pi_i f(x_i) \Phi(s-x_i).\] 

Hence, we have 

\[g(s,\pi)= \frac{\alpha(X,\pi)}{\alpha(1,\pi)}\]

and
\begin{align*}
|g(s,\pi)-g(s,\widetilde{\pi})| =& \left|\frac{\alpha(X,\pi)}{\alpha(1,\pi)} - \frac{\alpha(X,\widetilde{\pi})}{\alpha(1,\widetilde{\pi})}\right|\\
=&  \frac{\left| \alpha(X,\pi) \alpha(1,\widetilde{\pi})- \alpha(X,\widetilde{\pi})\alpha(1,\pi)\right|}{\alpha(1,\pi)\alpha(1,\widetilde{\pi})}\\
\leq&  \frac{\left| \alpha(X,\pi) \alpha(1,\widetilde{\pi})- \alpha(X,\pi) \alpha(1,\pi)\right|}{\alpha(1,\pi)\alpha(1,\widetilde{\pi})}\\
&+\frac{\left|  \alpha(X,\pi) \alpha(1,\pi)-  \alpha(X,\widetilde{\pi}) \alpha(1,\pi)\right|}{\alpha(1,\pi)\alpha(1,\widetilde{\pi})}\\
=&  \frac{\left| \alpha(X,\pi) \sum_{i=1}^n (\widetilde{\pi}_i -\pi_i) \Phi(s-x_i)\right|}{\alpha(1,\pi)\alpha(1,\widetilde{\pi})}\\
&+\frac{\left| \alpha(1,\pi)\sum_{i=1}^n (\pi_i -\widetilde{\pi}_i) x_i \Phi(s-x_i) \right|}{\alpha(1,\pi)\alpha(1,\widetilde{\pi})}\\
\leq & L \sum_{i=1}^n|\pi_i-\widetilde{\pi}_i|.
\end{align*}

The last inequality follows from $x_i\Phi(s-x_i)  \leq x_{\max} < \infty$ and $\Phi(s-x_i)  \leq 1$ for all $i$ and $\alpha(1,\pi) \geq \Phi(s-x_{\min}) \geq \Phi(C) > 0$. Together with (\ref{fineq}) this proofs the lemma.
\end{proof} 

\begin{remark}
 It is easy to see, that the additional restrictions $1>\Phi(0)$ and 
\begin{equation}\label{condbid}
  -\Phi'(y)\leq \frac{K}{C}  (1-\Phi(y))
\end{equation}
for all $y\in [-C,C]$ and a constant $K< 1$ (together with the differentiability of $\Phi$) are those that we need to obtain unique static Glosten-Milgrom bid prices. In addition, denoting the analogon of $g$ by
\begin{equation}\label{defh}
 h(s,\pi):= E[X|X+\eps \leq s]
\end{equation}
Lemma \ref{lippi} also holds for the bid price. Together with the assumptions on $\Phi$ in Lemma \ref{allwbuy} (\ref{condbid}) results in the assumptions on $\Phi$ in Theorem \ref{main}.  
\end{remark}

\section{Proof of Theorem \ref{main}}\label{proof}
To proof existence and uniqueness of a solution of a GMPS we firstly characterize it as a fixed point of a functional $F$ (see Definition \ref{deff}) that is defined on the set of admissible pricing strategies (see Theorem \ref{fixsol}). Then, we show a contraction of $F$ (see Lemma \ref{contr}) which can be used to verify uniqueness. Finally, we show that $F$ possesses a fixed point (as $F$ does in general not map into the set of admissible strategies and the contraction holds in general only if the arguments are admissible strategies, we cannot use a Picard-iteration).

\subsection{Glosten-Milgrom strategies as fixed points}\label{21.9.2012.1}
As we mentioned earlier the filtration of the market maker $\bbf^S$ does not satisfy the usual conditions, since it does not contain all null sets. We now define the completion $\widetilde{\bbf}^S$ of $\bbf^S$.

\begin{definition}
For any $\F\otimes\B([0,\infty))$-measurable process $S=(\ask,\bid)$ let the filtration $\widetilde{\bbf}^S$ be defined by 

\[\widetilde{\F}_t^S := \F_t^S \vee \N,\]

where $\N$ are all $P$-null sets of $\F$.
\end{definition}

$\widetilde{\bbf}^S$ will be used in the proof, but note that it is not needed to state our main result, Theorem \ref{main}.

\begin{lemma}
For any $\F\otimes\B([0,\infty))$-measurable process $S=(\ask,\bid)$, there exists a unique (up to indistinguishability) $\widetilde{\bbf}^S$-adapted \cdl process $\pi^S$ with
\begin{equation} \label{filter}
 \pi_{\tau}^S = \left(P\left[X_{\tau}=x_i| \F_{\tau}^S\right]\right)_{i=1,\ldots,n} \qquad \mbox{$P$-a.s.}
\end{equation}
for all finite stopping times $\tau$. 
\end{lemma}
\begin{proof}
Since $\widetilde{\bbf}^S$ satisfies the usual conditions, we can apply Theorems 2.7 and 2.9 of  \cite{bain2008fundamentals} to the process $\left(1_{\{X_t=x_i\}}\right)_{i=1,\ldots,n}$, which gives us a \cdl optional projection $\pi^S$ that is $\widetilde{\bbf}^S$-adapted and satisfies 
\[\pi_{\tau}^S = \left(P\left[X_{\tau}=x_i| \widetilde{\F}_{\tau}^S\right]\right)_{i=1,\ldots,n} \qquad \mbox{$P$-a.s.}\]
for all finite stopping times $\tau$. Since $E[\,\cdot\,|\F_\tau]$ and $E[\,\cdot\,|\widetilde{\F}_\tau]$ only differ by a $P$-null set, (\ref{filter}) follows. 
\end{proof}

Since $\pi^S$ is \cdl, $\pi^S_{t-}$ is well defined and we can now define the before mentioned functional.

\begin{definition}\label{deff}
 For admissible $S$ we define $F(S):\Omega \times [0,\infty) \to \bbr^2$ by 

\[F(S)_t :=  \left(\overline{F(S)}_t,\underline{F(S)}_t\right):=\left(g\left(\ask_t, \pi^S_{t-}\right),h\left(\bid_t, \pi^S_{t-}\right)\right)\]
where $g$ is defined in Definition \ref{defgh}, $h$ in (\ref{defh}), and $\pi^S$ in (\ref{filter}).

\end{definition}

As a continuous function of $\widetilde{\bbf}^S$-predictable processes $F(S)$ is $\widetilde{\bbf}^S$-predictable. By definition of $g$ and $h$ and the fact that $S$ is admissible it follows that $\overline{F(S)}_t(\omega)\geq \underline{F(S)}_t(\omega)$ for all $(\omega,t)$. However, (also not after completion of the filtrations) $F(S)$ is not necessarily admissible, since in general $\bbf^S \neq \bbf^{F(S)}$. 

\begin{figure}[htbp]
\begin{center}
 \includegraphics[scale=0.65]{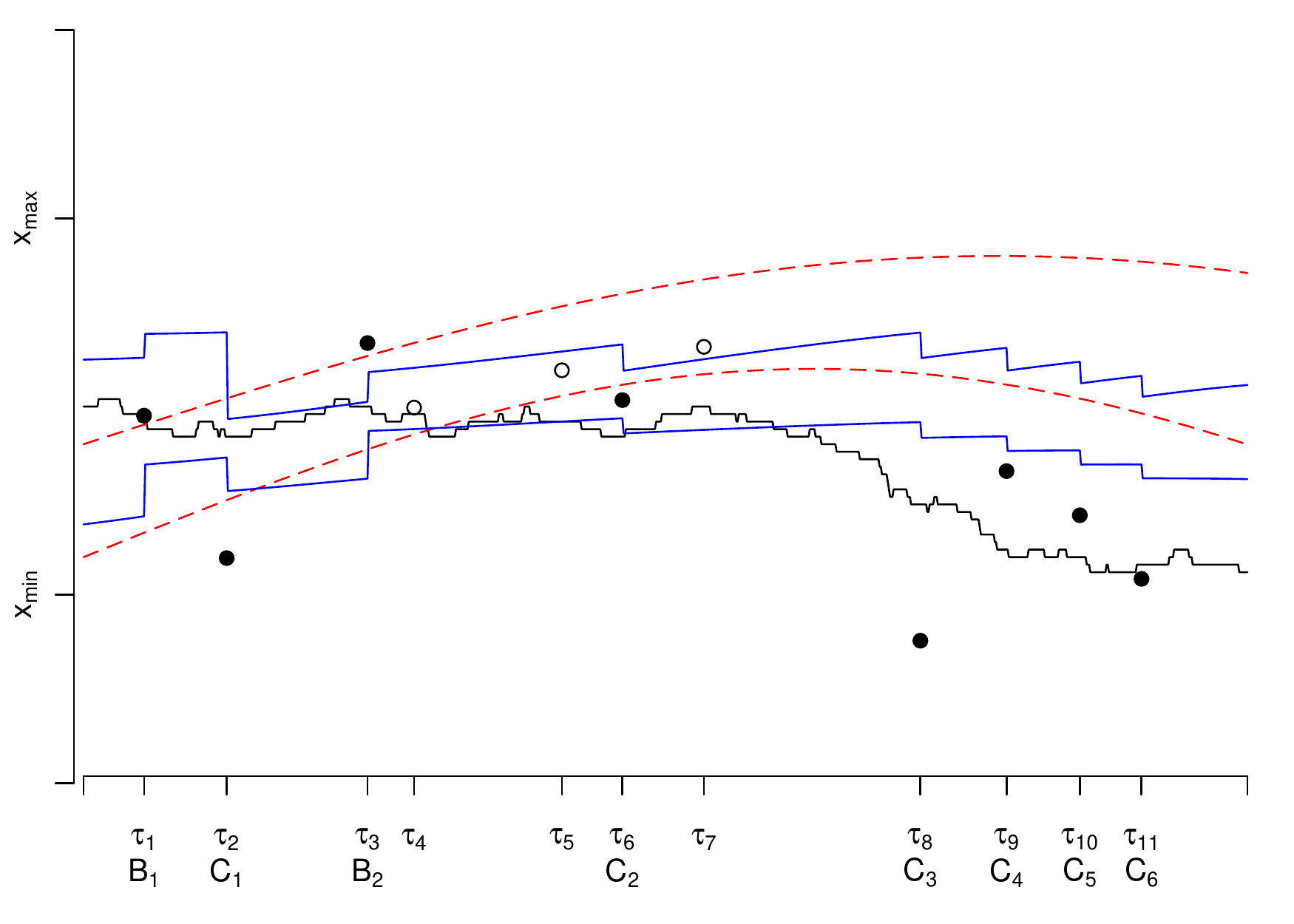}
\end{center}
\caption{We add $F(S)$ to Figure \ref{fig1} which are the fictitious Glosten-Milgrom-prices (i.e. zero-expected-profits) of the market maker if actually prices $S$ are quoted and the market reacts with buys and sells to them.}
\label{fig2}
\end{figure}

We now define a larger filtration  $\bbf= (\F_t)_{t\geq 0}$ by

\[\F_t:= \sigma \left(X_s, N_s, \eps_i 1_{\{\tau_i \leq s\}}, s\leq t, i\in \bbn \right) \vee \N, \]

which contains all information up to time $t$. The following lemma describes the intensity of the jump process $N^B$ that counts actual buys (for a given pricing strategy $S$) as described in Section \ref{model}.  

\begin{lemma} \label{rate}
 The $\bbf$-intensity of $N^B$ (in the sense of \cite{bremaud1981point} II, D7) is given by $\rate \Phi(\ask-X_-)$.
\end{lemma}
\begin{proof}
Let $C$ be a nonnegative $\bbf$-predictable process. Then 
\begin{align*}
E\left[\int_0^\infty C_s dN^B_s\right] &= E\left[\sum_{i=1}^{\infty} C_{\tau_i} 1_{\{X_{\tau_i}+\eps_i\geq \ask_{\tau_i}\}}\right]\\
&= \sum_{i=1}^{\infty} E\left[ E\left[C_{\tau_i} 1_{\{X_{\tau_i}+\eps_i\geq \ask_{\tau_i}\}}|\F_{\tau_i-}\right]\right]\\
&= \sum_{i=1}^{\infty} E\left[ C_{\tau_i} \Phi(\ask_{\tau_i}-X_{\tau_i-})\right]\\
&= E\left[\int_0^\infty C_s \Phi(\ask_s-X_{s-})dN_s\right]\\
&= E\left[\int_0^\infty C_s \rate \Phi(\ask_s-X_{s-})ds\right]
\end{align*}
where we use $X_{\tau_i}=X_{\tau_i-}$ $P-a.s.$ for the third equation.
\end{proof}

We now derive the filter equation of $\pi^S$.
\begin{lemma}\label{filtereq}
The process $\pi^S$ satisfies the following SDE 
\begin{equation*}
 \begin{aligned}
 d\pi_t^{S,i}  =& \pi_{t-}^{S,i} \left(\frac{\Phi(\ask_t-x_i)}{\sum_{j} \pi_{t-}^{S,j} \Phi(\ask_t-x_j)} -1\right) dN^B_t\\ 
& +\pi_{t-}^{S,i} \left(\frac{\Psi(\bid_t-x_i)}{\sum_{j} \pi_{t-}^{S,j} \Psi(\bid_t-x_j)} -1\right) dN^C_t \\
&- \left(\rate \pi_t^{S,i} \left(\vphantom{\sum_{j}}\Psi(\bid_t-x_i)+\Phi(\ask_t-x_i)\right. \right.\\
&\left.-\sum_{j} \pi_{t}^{S,j} \left(\Psi(\bid_t-x_j)+ \Phi(\ask_t-x_j)\right)\right)\\
&- \left.\sum_{j} \pi_t^{S,j} q(j,i)\right) dt.
\end{aligned}
\end{equation*}
for all $t\geq 0$, up to indistinguishability, with initial condition $\pi_0^{S,i}= P[X_0=x_i]$, where $\Psi(x) = P[\eps_1 \leq x]$.
\end{lemma}
 
\begin{proof}
We can derive the filter equation for  $\pi^S$ as it is done in \cite{bremaud1981point} IV, T2. The filter equation has the form

\[ \pi^S_t = \pi^S_0 + \int_0^t K^B_s  dN^B_s+ \int_0^t K^C_s dN^C_s +\int_0^t\left(- K^B_s \widehat{\lam}^B - K^C_s \widehat{\lam}^C +f_s\right) ds,\] 

where $\widehat{\lam}^B$ and $\widehat{\lam}^C$ are the $\widetilde{\bbf}^S$-intensities of $N^B$ and $N^C$ respectively, $f$ is the $\widetilde{\bbf}^S$-compensator of $X$, which is given by $\sum_{j} \pi^{S,j} q(j,\cdot)$ here. $K^B_s$ is described by 
the theorem as $\Psi^B_s- \pi^S_{s-}$ and $K^C_s=\Psi^C_s- \pi^S_{s-}$ respectively, where $\Psi^B_s$ is the unique (up to a 
$(P \otimes \lam)$-null set) $\widetilde{\bbf}^S$-predictable process satisfying
\begin{equation}\label{remtsh}
E\left[\int_0^t C_s 1_{\{X_s=x_i \}}\lam^B_s ds\right]= E\left[\int_0^t C_s \Psi^{B,i}_s \widehat{\lam}^B_s ds\right]
\end{equation} 
for all $\widetilde{\bbf}^S$-predictable nonnegative bounded processes $C$, $i=1,\ldots,n$ and all $t \geq 0$ where $\lam^B, \widehat{\lam}^B$ are the $\bbf,\widetilde{\bbf}^S$-intensities of $N^B$ respectively. A similar equation holds for $\Psi^C_s$. 

>From Lemma \ref{rate} we have that $\lam^B= \rate \Phi(\ask-X_-)$ and 

\begin{align*}
 E\left[\lam^B|\F_s^S\right]&= \sum_{i=1}^n E\left[1_{\{X_{s-}=x_i \}}\rate \Phi(\ask-x_i) |\F_s^S\right]\\
&= \sum_{i=1}^n \rate \Phi(\ask-x_i)E\left[1_{\{X_{s-}=x_i \}} |\F_s^S\right] \mbox{ $P$-a.s.}. 
\end{align*}

Since, for fixed $s\in \bbr_+$, $X_s=X_{s-}$ $P$-a.s., $\rate \sum_{i=1}^n \pi^{S,i}_{s-} \Phi(\ask_s-x_i)$ is a version of $\widehat{\lam}^B_s$. From this and as $\pi^{S,i}_-$ is the $\widetilde{\bbf}^S$-compensator of $1_{\{X=x_i\}}$ it follows that 

\[\Psi^{B,i}_s = \frac{ \pi^{S,i}_{s-}\Phi(\ask_s-x_i)}{\sum_{j=1}^n \pi^{S,j}_{s-} \Phi(\ask_s-x_j)}\]

solves (\ref{remtsh}) which gives us $K^B_s$ and analog $K^C_s$ as stated in the lemma.

For the buy-side part of the $dt$-term we get

\[- K^{B,i}_s \widehat{\lam}^B= -\left(\frac{ \pi^{S,i}_{s-}\Phi(\ask_s-x_i)}{\sum_{j=1}^n \pi^{S,j}_{s-} \Phi(\ask_s-x_j)}-\pi^{S,i}_s\right) \rate \sum_{j=1}^n \pi^{S,j}_{s-} \Phi(\ask_s-x_j)\]

which simplifies to

\[-\rate\pi^{S,i}_{s-}\Phi(\ask_s-x_i)+\rate\pi^{S,i}_{s-} \sum_{j=1}^n \pi^{S,j}_{s-} \Phi(\ask_s-x_j).\]

Together with $f$ and the similar results for the sell-side we obtain the $dt$-term as stated in the Lemma which completes the proof.
\end{proof}

We now can prove the following Lemma.

\begin{lemma}\label{fatb} We have
\[\overline{F(S)}_{B_i} = E[X_{B_i}| \F^S_{B_i}] \qquad \mbox{and} \qquad\underline{F(S)}_{C_i} = E[X_{C_i}| \F^S_{C_i}]\qquad \mbox{$P$-a.s.}\]
for all $i\in \bbn$. 
\end{lemma}
\begin{proof}

By the filter equation in Lemma \ref{filtereq} and the fact that $N^B$ and $N^C$ have no common jumps it follows that 

\begin{equation}\label{pieqpsi}
\pi^{S,k}_{B_i}=\frac{\pi^{S,k}_{B_i-} \Phi(\ask_{B_i}-x_k)}{\sum_{j=1}^n \pi^{S,j}_{B_i-} \Phi(\ask_{B_i}-x_j)} , \qquad \mbox{for } k =1, \ldots,n.
\end{equation}

By definition of $F$ and $g$ and by (\ref{pieqpsi}) we have
\begin{align*}
\overline{F(S)}_{B_i} &= g\left(\ask_{B_i}, \pi^S_{B_i-}\right)= \frac{\sum_{j=1}^n x_j \pi^{S,j}_{B_i-} \Phi(\ask_{B_i}-x_j)}{\sum_{j=1}^n \pi^{S,j}_{B_i-} \Phi(\ask_{B_i}-x_j)}\\
&= \sum_{j=1}^n x_j \pi^{S,j}_{B_i} =E[X_{B_i}| \F^S_{B_i}].
\end{align*}
The same holds true at the times when a sell occurs.
\end{proof}

The importance of that assertion becomes clear if we put it together with the next lemma that describes an equivalent and maybe more intuitive description of the criterion of the GMPS in Definition \ref{admis}.

\begin{lemma}\label{equidef}
$S$ is a GMPS iff it is admissible and 
\[\ask_{B_i}=E[X_{B_i}|\F^S_{B_i}] \qquad \mbox{and} \qquad \bid_{C_i} = E[X_{C_i}|\F^S_{C_i}]\qquad P-a.s.\]
for all $i\in \bbn$.
\end{lemma}

This means that all trades in a GMPS are executed at a price which is the expectation of $X$ given the information available to market makers at that point of time. The interesting point about this characterization of GMPS is that a trade which occurs at that very moment is included in the filtration but its occurrence and especially its direction is not predictable in contrast to $S$.

\begin{remark}
 In view of Lemma \ref{fatb} and Lemma \ref{equidef}, we can give an intuitive description of $F$. $F(S)$ are the Glosten-Milgrom-prices (i.e. zero-expected-profits) a market maker would have in mind if actually the prices $S$ are quoted and the market reacts with buys and sells to them (which leads to the filtration $\bbf^S$).
\end{remark}

\begin{proof}[Proof of Lemma \ref{equidef}]
We only consider buys. Let $S$ be a GMPS. For fixed $i \in \bbn$ and $t \in \bbr_+$ we consider 
\[C:=\{B_{i-1} \leq t < B_i\}\cap A\]
for $A\in \F^S_t$. For $n \in \bbn, n>t$, let $\kappa_n(\omega) := t$ if $\omega \notin C$ and $\kappa_n(\omega) := B_i(\omega)\wedge n$ if $\omega \in C$. Hence $t\leq \kappa_n$ for all $n$ and both are bounded $\bbf^S$-stopping times. Thus, we have
\begin{align*}
0 &= E\left[\sum_{B_j \leq \kappa_n} (\ask_{B_j}-X_{B_j})\right] - E\left[\sum_{B_j \leq t} (\ask_{B_j}-X_{B_j})\right]\\
&= E\left[\sum_{t<B_j \leq \kappa_n} (\ask_{B_j}-X_{B_j})1_C\right] = E\left[(\ask_{B_i}-X_{B_i})1_{C\cap\{B_i \leq n\}}\right]
\end{align*}
for all $n$ and therefore 

\[ E\left[ (\ask_{B_i}-X_{B_i}) 1_C\right] = 0.\]

Note that  $\F_{B_i-}^S$ is generated by sets of the form $\{t < B_i\}\cap A$, where $A\in \F^S_t$ and $t \in \bbr_+$which can be writen as

\[\{t < B_i\}\cap A = \cup_{t \leq t_n \in \bbq} \{B_{i-1} \leq t_n < B_i\}\cap A.\]

Hence, as $A \in \F_t \subset \F_{t_n}$ for $t_n \geq t$  and as $\{B_{i-1}\leq t\}\in \F^S_t$, the sets of the form like $C$ generate $\F_{B_i-}^S$. Since in addition $\ask$ is predictable it follows that 

\begin{equation}\label{seqexbifm}
\ask_{B_i} =E\left[ X_{B_i}| \F_{B_i-}^S\right] \qquad  \mbox{ $P$-a.s.}.
\end{equation}

Let us show that $\F_{B_i-}^S=\F_{B_i}^S$. We consider the following marked point process $(T_n, Z_n)_{n \in \bbn}$, where 

\[T_n := \inf\left\{t \ |\ \sum_{i\geq1} 1_{\{B_i\leq t\}}+ \sum_{i\geq1} 1_{\{C_i\leq t\}}\geq n\right\},\] 

which are the times of trades (buys and sells) and $Z_n :=1$ if $T_n = B_i$ for some $i$ and $Z_n:=-1$ if $T_n = C_i$. Note that the proof only works as buys and sells never happen simultaneously.

We can now write 

\begin{equation}\label{fsbi}
 \F_{B_i}^S = \{A \ |\ A = \cup_{n\in \bbn} A_n \cap\{B_i=T_n\}\mbox{ for } A_n \in \F_{T_n}^S \mbox{ for all } n\}
\end{equation}

and
\begin{equation}\label{fsbim}
\F_{B_i-}^S = \{A \ |\ A = \cup_{n\in \bbn} A_n \cap\{B_i=T_n\}\mbox{ for } A_n \in \F_{T_n-}^S \mbox{ for all } n\}.
\end{equation}

The first equation can be seen easily. For ''$\subset$`` of the second it suffices to show that sets of the form $A\cap \{t < B_i\}$, $A \in \F^S_t$ are in the set on the RHS. This can be done by choosing $A_n = A \cap \{t < T_n\}$. For ''$\supset$`` it again suffices to consider sets of the form $A_n = \widetilde{A}_n \cap \{t < T_n\}$, $\widetilde{A}_n \in \F^S_t$. It then remains to show that $\{B_i=T_n\} \in \F_{B_i-}^S$. However, 

\[\{T_n< B_i\} = \cup_{q \in \bbq} \{T_n< q\} \cap \{q < B_i\} \in \F_{B_i-}^S \qquad \mbox{for all $n \in \bbn$ and thus}\]

\[\{B_i=T_n\} = \{T_{n-1}< B_i\} \cap \{T_n< B_i\}^c\in \F_{B_i-}^S \qquad \mbox{for all $n \in \bbn$}.\]

Now, by Theorem III, T2 in \cite{bremaud1981point} applied to the marked point process $(T_n, Z_n)_{n \in \bbn}$ any $A_n \in \F^S_{T_n}$ can be written as 

\[A_n = (M_1 \cap \{Z_n=1\}) \cup (M_2 \cap \{Z_n=-1\})\]

for $M_1,M_2 \in \F_{T_n-}^S$. Since for fixed $i$ $\{B_i=T_n\}=\{Z_n=1\}$, we have 

\[A_n \cap\{B_i=T_n\} = M_1 \cap \{B_i=T_n\} \]
 and  $\F_{B_i-}^S=\F_{B_i}^S$ follows from (\ref{fsbi}) and (\ref{fsbim}). Together with (\ref{seqexbifm}) one direction of the lemma is proven.

Now let $\ask_{B_i}=E[X_{B_i}|\F^S_{B_i}]$ for all $i\in \bbn$ and $\tau$ be a bounded $\bbf^S$-stopping time. We obtain that
\begin{align*}
E\left[\sum_{B_i \leq \tau} (\ask_{B_i}-X_{B_i})\right]&= \sum_{i=1}^\infty E\left[1_{\{B_i \leq \tau\}} (\ask_{B_i}-X_{B_i})\right]\\
&= \sum_{i=1}^\infty E\left[E\left[1_{\{B_i \leq \tau\}} (\ask_{B_i}-X_{B_i})|\F^S_{B_i}\right]\right]\\
&= \sum_{i=1}^\infty E\left[1_{\{B_i \leq \tau\}} (\ask_{B_i}-E\left[X_{B_i}|\F^S_{B_i}\right])\right]\\
&= 0.\qedhere
\end{align*}
\end{proof}

\begin{definition}
 We say that an admissible strategy $S$ is a fixed point of $F$, if $S= F(S)$ $(P\otimes \lam)$-a.e. (where $\lam$ denotes the Lebesgue measure on $\bbr_+$).
\end{definition}

\begin{theorem} \label{fixsol}
An admissible strategy $S$ is a solution of the GMPS-problem iff $S$ is a fixed point of $F$.
\end{theorem}

\begin{proof}
Identity of predictable sets $B$ up to a $(P\otimes \lam)$-null set is equivalent to equality at Poisson times, since
\[(P\otimes\lam) (B) = \sum_{j=1}^\infty P[(\omega,\tau_j(\omega))\in B]\]

Let $S$ be a fixed point of $F$. With Lemma \ref{fatb} we obtain
\[\ask_{B_i} = \overline{F(S)}_{B_i} = E[X_{B_i}| \F^S_{B_i}]\mbox{ and } \bid_{C_i} = \underline{F(S)}_{C_i} = E[X_{C_i}| \F^S_{C_i}]\qquad \mbox{$P$-a.s.}\]
for all $i\in \bbn$. Lemma \ref{equidef} now yields that $S$ is a GMPS. 

Now let $S$ be a GMPS and $j\in \bbn$ be fixed. We consider the set 

\[A:=\left\{ \ask_{\tau_j} \neq \overline{F(S)}_{\tau_j}\right\}\]

which is in $\F_{\tau_j -}$ since $S$ and $F(S)$ are $\bbf^S$-predictable and hence $\bbf$-predictable.

Since $S$ is a GMPS and $\eps_j$ is independent of $\F_{\tau_j -}$  we have by Lemma \ref{equidef}

\begin{align*}
0&= P\left[\ask_{B_i}\neq \overline{F(S)}_{B_i} \mbox{ for some } i\in\bbn\right]\\
 &\geq P[\{\tau_j=B_i \mbox{ for some } i \in \bbn \}\cap A] = P[\{X_{\tau_j} + \eps_j \geq \ask_{\tau_j}\} \cap A]\\
 &\geq  P[\{ \eps_j \geq x_{\max}-x_{\min}\} \cap A] = P[A]P[\eps_j\geq  x_{\max}-x_{\min}].
\end{align*}

Since by Lemma \ref{allwbuy} $P[\eps_j\geq  x_{\max}-x_{\min}]>0$ it follows that $P[A]=0$, in other words, for all $\tau_j$ we have

\[\ask_{\tau_j}  =\overline{F(S)}_{\tau_j} \qquad \mbox{$P$-a.s.}. \qedhere\]

\end{proof}

\subsection{Uniqueness}
We first show uniqueness of the solution by proving that $F$ is a contraction in the sense of Lemma \ref{contr}. Let $S$ and $T$ be admissible pricing strategies. 

\begin{definition}
For given pricing strategies $S,T$ let 
\[A^1_s:= \{ X_{\tau_i}+\eps_i \notin [\min\{\overline{S}_{\tau_i}, \overline{T}_{\tau_i}\},\max\{\overline{S}_{\tau_i}, \overline{T}_{\tau_i}\})\mbox{ for all } \tau_i \leq s\}. \]

\end{definition}

$A^1_s$ is the event that until $s$ no buy occurred only in one of the two pricing strategies $S$ and $T$. 
\begin{lemma}\label{lema} We have that
\[P\left[ (A^1_s)^c\right] \leq \rate M E \left[\int_0^s  |\overline{S}_u-\overline{T}_u|du\right]\]
where $M:=\max\{\Phi'(x)|x\in [-C,C]\}$.
\end{lemma}
\begin{proof}Let $Y$ be the process that counts the number of buys that only occur for one pricing strategy, i.e.

\[ Y_t:=\sum_{i\in \bbn} 1_{\left\{\tau_i\leq t, X_{\tau_i}+\eps_i \in [\min\{\overline{S}_{\tau_i}, \overline{T}_{\tau_i}\},\max\{\overline{S}_{\tau_i}, 
\overline{T}_{\tau_i}\})\right\}}.\]

We now show (essentially with the methods of the proof of Lemma \ref{rate}) that the $\bbf$-intensity of $Y$ is given by

\begin{align*}
\lam^Y_t:= &\rate (\Phi(\min\{\overline{S}_t, \overline{T}_t\}-X_{t-}) - \Phi(\max\{\overline{S}_t, \overline{T}_t \}-X_{t-}))\leq \rate M |\overline{S}_t-\overline{T}_t|.\\
\end{align*}

Let $C$ be a nonnegative $\bbf$-predictable process. As $\ask$  and $\overline{T}$ are $\bbf$-predictable and $P[X_{\tau_i}= X_{\tau_i-}]=1$, we obtain

\begin{align*}
E\left[\int_0^\infty C_s dY_s\right] &= E\left[\sum_{i=1}^{\infty} C_{\tau_i} 1_{\{\min\{\overline{S}_{\tau_i}, \overline{T}_{\tau_i}\}\leq X_{\tau_i}+\eps_i < \max\{\overline{S}_{\tau_i}, \overline{T}_{\tau_i} \}\}}\right]\\
&= \sum_{i=1}^{\infty} E\left[ E\left[C_{\tau_i}1_{\{\min\{\overline{S}_{\tau_i}, \overline{T}_{\tau_i}\}\leq X_{\tau_i}+\eps_i < \max\{\overline{S}_{\tau_i}, \overline{T}_{\tau_i} \}\}}|\F_{\tau_i-}\right]\right]\\
&= \sum_{i=1}^{\infty} E\left[ C_{\tau_i}\lam^Y_{\tau_i}\right]= E\left[\int_0^\infty C_s \lam^Y_s ds\right].
\end{align*}

We define $\tau_Y:=\inf\{t\geq 0\ |\ Y_t=1\}$ and get

\begin{align*}
P[(A^1_s)^c]&= P[Y_s \neq 0] = E\left[\int_0^s 1_{\{\tau_Y\geq u\}}dY_u\right] = E\left[\int_0^s 1_{\{\tau_Y\geq u\}}\lam^Y_u du\right]\\
&\leq E \left[\int_0^s \lam^Y_u du\right] \leq \rate M  E \left[\int_0^s|\overline{S}_u-\overline{T}_u|du\right].\qedhere
\end{align*}
\end{proof}

The same holds true for sells. Hence for 
\[A^2_s:= \{ X_{\tau_i}+\eps_i \notin (\min\{\underline{S}_{\tau_i}, \underline{T}_{\tau_i}\},\max\{\underline{S}_{\tau_i}, \underline{T}_{\tau_i}\}]\mbox{ for all } \tau_i \leq s\}\]
we obtain a similar estimate and for $A_s := A^1_s \cap A^2_s$, which is the event that the same buys and sells are observed in the two pricing strategies $S$ and $T$ we have 

\begin{equation} \label{aesti}
P\left[ A_s^c\right] \leq 2 \rate M E \left[\int_0^s  \norm{S_u-T_u}du  \right]\qquad\mbox{where}
\end{equation}

\[\norm{S_u-T_u}:= \max\left\{|\overline{S}_u-\overline{T}_u|,|\underline{S}_u-\underline{T}_u|\right\}.\]

\begin{lemma} \label{contr}
There is a constant $K_1 < \infty$ such that
\[E\left[\int_0^t \norm{F(S)_s - F(T)_s} ds \right] \leq (K + t K_1)  E\left[\int_0^t\norm{S_s-T_s} ds \right] \]
for all $t\geq0$ and for $K$ from Theorem \ref{main}.
\end{lemma}

\begin{proof}

First we estimate the difference of the conditional distributions of the true value resulting from different pricing strategies. We obtain

\begin{equation}\label{difcond}
 \begin{aligned}
E\left[ \left|\pi^{S,i}_s-\pi^{T,i}_s\right|  \right] &=E\left[ \left|P\left[X_s=x_i|\F_s^S\right]-P\left[X_s=x_i|\F_s^T\right]\right| \left(1_{A_s^c}+1_{A_s}\right)\right] \\
&\leq E\left[ 1_{A_s^c} \right] +E\left[ \left|E\left[1_{\{X_s=x_i\}}|\F_s^S\right]-E\left[1_{\{X_s=x_i\}}|\F_s^T\right]\right| 1_{A_s\}} \right]\\
& =P\left[ A_s^c\right] +E[ \left|E\left[1_{\{X_s=x_i\}} \left(1_{A_s^c}+1_{A_s}\right)\right|\F_s^S\right]\\
& \qquad-E\left[1_{\{X_s=x_i\}} \left(1_{A_s^c}+1_{A_s}\right)|\F_s^T\right] | 1_{A_s} ]\\
&\leq  3 P\left[ A_s^c\right] +E[ |E\left[1_{\{X_s=x_i\}} 1_{A_s}|\F_s^S\right]\\
&\qquad -E\left[1_{\{X_s=x_i\}} 1_{A_s}|\F_s^T\right] | 1_{A_s} ]\\
&= 3 P\left[ A_s^c\right], \qquad i =1, \ldots, n.
\end{aligned}
\end{equation}

The last equation holds true since $\F^S_s\cap A_s = \F^T\cap A_s$. The equality of the trace $\sigma$-algebras holds due to

\[\F^S_s\cap A_s= \sigma(\{B^S_i \leq u\}, \{C^S_i \leq u\}, u\leq s, i \in \mathbb{N})\cap A_s,\]

and obviously 

\[ \{B^S_i \leq u\}\cap A_s = \{B^T_i \leq u\}\cap A_s \mbox{ and }  \{C^S_i \leq u\}\cap A_s = \{C^T_i \leq u\}\cap A_s\]

respectively for all $i\in \bbn$ and $u\leq s$. Putting (\ref{aesti}) and (\ref{difcond}) together we obtain

\begin{align*}
E\left[\int_0^t \sum_{i=1}^n|\pi_s^{S,i}-\pi_s^{T,i}| ds \right] & = \int_0^t \sum_{i=1}^n E\left[|\pi_s^{S,i}-\pi_s^{T,i}|  \right]ds\\
&\leq \int_0^t  \sum_{i=1}^n    3 P\left[ A_s^c\right]ds\\
&\leq   6 n \rate M \int_0^t  E \left[\int_0^s  \norm{S_u-T_u}du \right]  ds\\
&\leq   6 n \rate M t E \left[\int_0^t  \norm{S_s-T_s} ds\right].  \\
\end{align*}
Finally, we have 

\begin{align*}
E\left[\int_0^t \left|\overline{F(S)}_s - \overline{F(T)}_s\right| ds \right] &= E\left[\int_0^t \left|g\left(\overline{S}_s,\pi_{s-}^S\right) - g\left(\overline{T}_s,\pi_{s-}^T\right) \right| ds \right]\\
&\leq E\left[\int_0^t K |\overline{S}_s-\overline{T}_s| + L \sum_{i=1}^n|\pi_s^{S,i}-\pi_s^{T,i}| ds \right]\\
&\leq E\left[\int_0^t K |\overline{S}_s-\overline{T}_s| +6 L  n \rate M t \norm{S_s-T_s} ds \right],
\end{align*}

where the first inequality is due to Lemma \ref{lippi} for $K<1$  defined in Theorem~\ref{main} and  $L = \frac{2 x_{\max}}{\Phi(C)^2}$. A similar estimate can be obtained for 
\[E\left[\int_0^t \left|\underline{F(S)}_s - \underline{F(T)}_s\right| ds \right].\]
We then get the desired result with $K_1= 12 L n \rate M$.
\end{proof}

\begin{proof}[Proof of uniqueness in Theorem \ref{main}]
Let $S,T$ be two solutions of the GMPS problem. By Theorem \ref{fixsol} every solution is a fixed point of $F$. Applying Lemma \ref{contr} with some $t>0$ s.t. $K+tK_1 < 1$ we obtain that $S$ and $T$ coincide 
$P\otimes\lam|_{[0,t]}$-a.e.. Note that $K$ and $K_1$ only depend on $x_{\min},x_{\max}$, the distribution of the $\eps_i$ and $\lam$, but it is independent of the probabilities $P[X_0=x_i]$. 

But if $S=T$ $P\otimes\lam|_{[0,t]}$-a.e. so are the $P$-completions of  $\F^S_t$ and $\F^T_t$. Iteratively it follows that $S$ and $T$ are equal on $[0,\infty)$ as all arguments from above hold true also for a non-trivial $\F_0$. 
\end{proof}

\subsection{Existence}

To show existence we will proceed as follows. We will define an $n$-dimensional process~$\phi$ as a pathwise solution of a stochastic integral equation, 
which is what we assume the conditional distribution of the true value under the filtration of a GMPS could look like. 
We then define prices as the static solutions for every $(\omega,t)$, plugging in the conditional distribution of the true value, and construct the corresponding market maker's filtration. 
Then, we show that, under the constructed filtration, $\phi$ is adapted and solves the filter equation of the conditional distribution of the true value. This shows with the results in 
Subsection~\ref{21.9.2012.1} that we have indeed constructed a GMPS.

\begin{definition}
Let $\phi \in [0,1]^n$ such that $\sum_{i=1}^n \phi_i=1$. With $G(\phi)$ and $H(\phi)$ we denote the unique solutions $s$ of
\[g(s,\phi)=s \mbox{ and } h(s,\phi)=s\]
respectively where $g$ and $h$ are defined in Definition \ref{defgh} and (\ref{defh}) respectively. The existence and uniqueness of that solution is secured by Theorem \ref{solstat}.
\end{definition}

In the following we still use the notation as before $\Phi(x)=P[\eps_1\geq x]$ and further denote the distribution function of the $\eps_i$ by $\Psi(x) = P[\eps_1 \leq x]$. 

\begin{proof}[Proof of existence in Theorem \ref{main}]

\textit{Step 1:} For  $\phi: \Omega \times  [0,\infty)\to [0,1]^n$ we consider the SDE 

\begin{equation}\label{phi}
 \begin{aligned}
 \phi_t^i  =& \;\phi_0^i +\sum_{\tau_k \leq t}  \phi_{\tau_k-}^i \left(\frac{\Phi(G(\phi_{\tau_k-})-x_i)}{\sum_{j} \phi_{\tau_k-}^j \Phi(G(\phi_{\tau_k-})-x_j)} -1\right)1_{\{X_{\tau_k}+\eps_k \geq G(\phi_{\tau_k-})\}}\\ 
& +\sum_{\tau_k \leq t} \phi_{\tau_k-}^i \left(\frac{\Psi(H(\phi_{\tau_k-})-x_i)}{\sum_{j} \phi_{\tau_k-}^j \Psi(H(\phi_{\tau_k-})-x_j)} -1\right)1_{\{X_{\tau_k}+\eps_k \leq H(\phi_{\tau_k-})\}}  \\
&- \int_0^t \left( \rate \phi_s^i \left(\vphantom{\sum_{j}}\Psi(H(\phi_s)-x_i)+\Phi(G(\phi_s)-x_i)\right. \right.\\
&\left.-\sum_{j} \phi_s^j \left(\Psi(H(\phi_s)-x_j)+ \Phi(G(\phi_s)-x_j)\right)\right)\\
&- \left. \sum_{j} \phi_s^j q(j,i)\right) \;ds
\end{aligned}
\end{equation}
with initial condition $\phi_0^i=P[X_0=x_i]$ for all $i =1, \ldots,n$. In a first step we consider this SDE only pathwise and show that it has a unique solution 
with \cdl paths (we do not yet have a filtration). We start by showing that $G$ (and $H$) are Lipschitz-continuous. 
Let $s,\widetilde{s}$ be such that $G(\phi)=s$, i.e. $g(s,\phi)=s$ and $G(\widetilde{\phi})=\widetilde{s}$. We have

\[|s-\widetilde{s}| = \left|g(s,\phi)-g(\widetilde{s},\widetilde{\phi})\right| \leq K |s-\widetilde{s}| + L \sum_{i=1}^n|\phi_i-\widetilde{\phi}_i|\]

by Lemma \ref{lippi}, where $K<1$ and $L< \infty$. By rearranging we get

\[|G(\phi)-G(\widetilde{\phi})|=|s-\widetilde{s}|\leq \frac{L}{1-K} \sum_{i=1}^n|\phi_i-\widetilde{\phi}_i|.\]

Further, the functions $\Phi$ and $\Psi$ are differentiable and the derivative is bounded by $\frac{K}{C}< \infty$  on the compact set $[-C,C]$.
In addition, $\Phi$ and $\Psi$ are bounded by one. 
By the product rule, it follows that the $ds$-term in (\ref{phi}) considered as a function in $\phi$ can be modified to a function $f(\phi)$ that is 
Lipschitz-continuous and $f$ coincides with the original function for all $\phi\in\bbr^n$ with $\phi^i\ge 0$ and $\sum_{i=1}^n \phi^i=1$. 
Then, the system of ordinary differential equations only consisting of the modified $ds$-terms has a unique solution and, by construction of the ODEs, 
the solution stays in the set of probabilities. 
Thus, it also solves the differential equations with the original $ds$-terms, i.e.
\[
\begin{aligned}
 d\psi_t^i =&-\left(\rate \psi_t^i \left(\vphantom{\sum_{j}}\Psi(H(\psi_t)-x_i)+\Phi(G(\psi_t)-x_i)\right. \right.\\
&\left.-\sum_{j} \psi_t^j \left(\Psi(H(\psi_t)-x_j)+ \Phi(G(\psi_t)-x_j)\right)\right)\\
&- \left. \sum_{j} \psi_t^j q(j,i) \right) dt.
\end{aligned}
\]

We can now construct a candidate for the original problem up to $\tau_1$ by this solution, i.e. $\phi_t:=\psi_t$ for all $t< \tau_1$, and 

\[
\begin{aligned}
 \phi^i_{\tau_1} = &\psi^i_{\tau_1-} +  \psi_{\tau_1-}^i \left(\frac{\Phi(G(\psi_{\tau_1-})-x_i)}{\sum_{j} \psi_{\tau_1-}^j \Phi(G(\psi_{\tau_1-})-x_j)} -1\right)1_{\{X_{\tau_1}+\eps_1 \geq G(\psi_{\tau_1-})\}}\\ 
& +\psi_{\tau_1-}^i \left(\frac{\Psi(H(\psi_{\tau_1-})-x_i)}{\sum_{j} \psi_{\tau_1-}^j \Psi(H(\psi_{\tau_1-})-x_j)} -1\right)1_{\{X_{\tau_1}+\eps_1 \leq H(\psi_{\tau_1-})\}}. 
\end{aligned}
\]

We also obtain a solution $\widetilde{\psi}$ of the ordinary differential equation above for every state of $\phi^i_{\tau_1}$ as initial condition. Considered as a parameter-depending 
differential equation, the solution is continuous in the initial condition.
We then define a solution  of the original problem on $(\tau_1, \tau_2)$ by

\[\phi_t= \widetilde{\psi}_{t-\tau_1}\]

and so on. Iteratively we obtain a process that satisfies (\ref{phi}) up to all $\tau_i$. Then, one may define 
$\phi^i_t(\omega)=1/n$ for $t\in\bbr_+$ with $t\ge \sup_{i\in\bbn}\tau_i(\omega)$. 
As $\tau_i$ are Poisson times, 
this definition of course only affects a $P$-null set, but the construction ensures measurability (see Step~2) without needing the usual conditions and without the additional 
assumption that $\sup_{i\in\bbn}\tau_i(\omega)=\infty$ for all $\omega\in\Omega$. The process $\phi: \Omega \times \bbr_+ \mapsto \bbr^n$ 
has \cdl paths  at least on $[0,\sup_{i\in\bbn}\tau_i(\omega))$.

\textit{Step 2:} We now define $S_t:=(G(\phi_{t-}),H(\phi_{t-}))$ on $(0,\sup_{i\in\bbn}\tau_i(\omega))$ (and maybe $S=(x_{\max},x_{\min})$ elsewhere)
and with it $N^B, N^C$ (with jump times $B_k$ resp. $C_k$) 
and $\bbf^S$ according to (\ref{defN}) and (\ref{deffs}) resp. 
It follows that the jumps in (\ref{phi}) only take place at actual buys and sells with prices $S$. 
Therefore and by the construction of $\phi$ (using the continuity in the initial condition), 
for every $t\in\bbr$, $\phi_t$ can be written as a measurable function
of $B_k 1_{\{B_k\le t\}}$ and $C_k 1_{\{C_k\le t\}}$, $k\in\bbn$. Thus $\phi_t$ is  $\mathcal{F}^S_t$-measurable, i.e. $\phi$ is  $\bbf^S$-adapted. 
It follows that the process $S$ that is left-continuous on $(0,\sup_{i\in\bbn}\tau_i(\omega))$
is $\bbf^S$-predictable and hence admissible in the sense 
of Definition~\ref{admis}. Note that by the pathwise construction we obtain pricing strategies that are $\bbf^S$-predictable and not only predictable w.r.t. the completed 
filtration $\widetilde{\bbf}^S$ that satisfies the usual conditions. By (\ref{defN}) we can write (\ref{phi}) as

\begin{equation}\label{phi2}
 \begin{aligned}
 d\phi_t^i  =& \phi_{t-}^i \left(\frac{\Phi(\ask_t-x_i)}{\sum_{j} \phi_{t-}^j \Phi(\ask_t-x_j)} -1\right) dN^B_t\\ 
& +\phi_{t-}^i \left(\frac{\Psi(\bid_t-x_i)}{\sum_{j} \phi_{t-}^j \Psi(\bid_t-x_j)} -1\right) dN^C_t \\
&- \left( \rate \phi_t^i \left(\vphantom{\sum_{j}}\Psi(\bid_t-x_i)+\Phi(\ask_t-x_i)\right. \right.\\
&\left.-\sum_{j} \phi_{t}^j \left(\Psi(\bid_t-x_j)+ \Phi(\ask_t-x_j)\right)\right)\\
&- \left. \sum_{j} \phi_t^j q(j,i) \right) dt.
\end{aligned}
\end{equation}

\begin{figure}[htbp]
\begin{center}
 \includegraphics[scale=0.65]{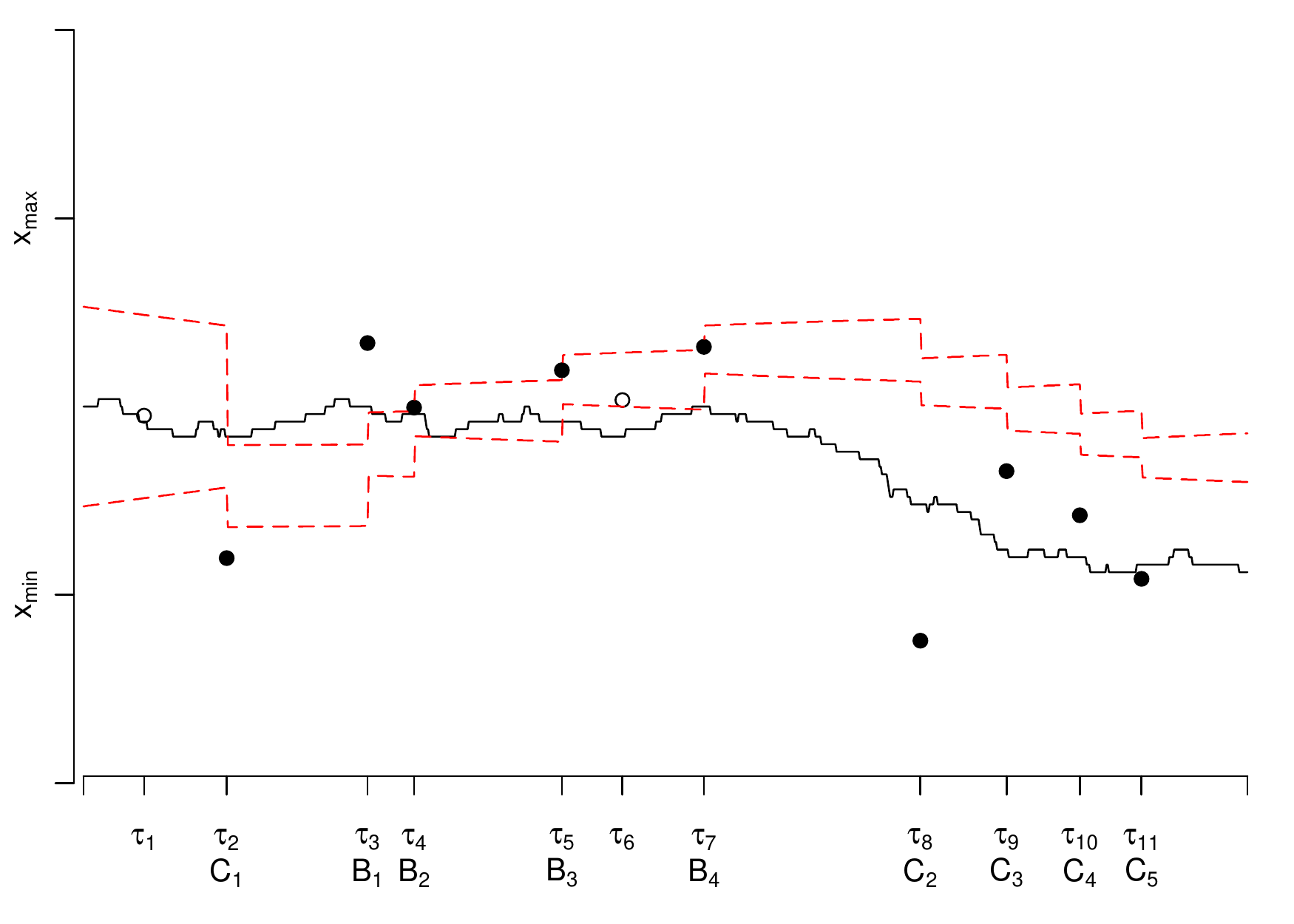}
\end{center}
\caption{Glosten-Milgrom-prices for the same scenario $\omega$ as in Figure \ref{fig1}.}
\label{fig3}
\end{figure}

\textit{Step 3:} We  described the filter equation of $\pi_t^{S,i} = P\left[X_t=x_i| \F_{\tau}^S\right]$ in  Lemma~\ref{filtereq}. It is given by
\begin{equation}\label{filtereq2}
 \begin{aligned}
 d\pi_t^{S,i}  =& \pi_{t-}^{S,i} \left(\frac{\Phi(\ask_t-x_i)}{\sum_{j} \pi_{t-}^{S,j} \Phi(\ask_t-x_j)} -1\right) dN^B_t\\ 
& +\pi_{t-}^{S,i} \left(\frac{\Psi(\bid_t-x_i)}{\sum_{j} \pi_{t-}^{S,j} \Psi(\bid_t-x_j)} -1\right) dN^C_t \\
&- \left(\rate \pi_t^{S,i} \left(\vphantom{\sum_{j}}\Psi(\bid_t-x_i)+\Phi(\ask_t-x_i)\right. \right.\\
&\left.-\sum_{j} \pi_{t}^{S,j} \left(\Psi(\bid_t-x_j)+ \Phi(\ask_t-x_j)\right)\right)\\
&+\left. \sum_{j} \pi_t^{S,j} q(j,i) \right) dt.
\end{aligned}
\end{equation}

Note that $S$ depends on $\phi$ and is fixed in (\ref{filtereq2}). In terms of $\pi^S$ (\ref{filtereq2}) has a unique solution  
and $\phi$ is obviously a solution of this equation (uniqueness follows as the $dt$-term considered as a function of $\pi^S$ is Lipschitz-continuous,
thus the arguments are similar but simpler as for (\ref{phi})). 
Thus, as $\phi$ and $\pi^S$ are both \cdl, they are indistinguishable. We then get for the ask price

\[\overline{F(S)}_t= g\left(\ask_t,\pi^S_{t-}\right) = g\left(G(\phi_{t-}),\phi_{t-}\right) = G(\phi_{t-}) = \ask_t,\]

also up to indistinguishability and thus $(P\otimes \lambda)$-a.e. As the same holds for the bid price Theorem~\ref{fixsol} shows that $S$ is a GMPS and 
existence is proven.
\end{proof}

\bibliography{references}{}

\begin{thebibliography}{BHMB{\O}12}

\bibitem[AB{\O}12]{aase.bjuland.oksendal.2012}
K.K. Aase, T.~Bjuland, and B.~{\O}ksendal.
\newblock Strategic insider trading equilibrium: a filter theory approach.
\newblock {\em Afrika Matematika}, 23(2):145--162, 2012.

\bibitem[AS08]{avellaneda2008high}
M.~Avellaneda and S.~Stoikov.
\newblock High-frequency trading in a limit order book.
\newblock {\em Quantitative Finance}, 8(3):217--224, 2008.

\bibitem[Bac92]{back1992insider}
K.~Back.
\newblock Insider trading in continuous time.
\newblock {\em Review of Financial Studies}, 5(3):387--409, 1992.

\bibitem[BB04]{back2004information}
K.~Back and S.~Baruch.
\newblock Information in securities markets: Kyle meets {G}losten and
  {M}ilgrom.
\newblock {\em Econometrica}, 72(2):433--465, 2004.

\bibitem[BC08]{bain2008fundamentals}
A.~Bain and D.~Crisan.
\newblock {\em Fundamentals of stochastic filtering}.
\newblock Springer Verlag, 2008.

\bibitem[BHMB{\O}12]{BHMO2012}
F.~Biagini, Y.~Hu, T.~Meyer-Brandis, and B.~{\O}ksendal.
\newblock Insider trading equilibrium in a market with memory.
\newblock {\em Mathematics and Financial Economics}, 6(3):229--247, 2012.

\bibitem[Br{\'e}81]{bremaud1981point}
P.~Br{\'e}maud.
\newblock {\em Point processes and queues, martingale dynamics}.
\newblock Springer, 1981.

\bibitem[BVH01]{BVH2001}
M.~Bagnoli, S.~Viswanathan, and C.~Holden.
\newblock On the existence of linear equilibria in models of market making.
\newblock {\em Mathematical Finance}, 11(1):1--31, 2001.

\bibitem[CG83]{copeland1983information}
T.E. Copeland and D.~Galai.
\newblock Information effects on the bid-ask spread.
\newblock {\em Journal of Finance}, pages 1457--1469, 1983.

\bibitem[CJ12]{cartea2012risk}
{\'A}.~Cartea and S.~Jaimungal.
\newblock Risk metrics and fine tuning of high frequency trading strategies.
\newblock {\em Mathematical Finance, forthcoming}, 2012.

\bibitem[Das05]{das2005learning}
S.~Das.
\newblock A learning market-maker in the {G}losten--{M}ilgrom model.
\newblock {\em Quantitative Finance}, 5(2):169--180, 2005.

\bibitem[Das08]{das2008effects}
S.~Das.
\newblock The effects of market-making on price dynamics.
\newblock In {\em Proceedings of the 7th international joint conference on
  autonomous agents and multiagent systems-Volume 2}, pages 887--894.
  International Foundation for autonomous Agents and Multiagent Systems, 2008.

\bibitem[GM85]{glosten1985bid}
L.R. Glosten and P.R. Milgrom.
\newblock Bid, ask and transaction prices in a specialist market with
  heterogeneously informed traders.
\newblock {\em Journal of Financial Economics}, 14(1):71--100, 1985.

\bibitem[GP12]{guilbaud2012optimal}
F.~Guilbaud and H.~Pham.
\newblock Optimal high frequency trading with limit and market orders.
\newblock {\em Quantitative Finance, forthcoming}, 2012.

\bibitem[HS81]{ho1981optimal}
T.~Ho and H.R. Stoll.
\newblock Optimal dealer pricing under transactions and return uncertainty.
\newblock {\em Journal of Financial Economics}, 9(1):47--73, 1981.

\bibitem[Kri92]{krishnan1992equivalence}
M.~Krishnan.
\newblock An equivalence between the {K}yle (1985) and the {G}losten--{M}ilgrom
  (1985) models.
\newblock {\em Economics Letters}, 40(3):333--338, 1992.

\bibitem[Kyl85]{kyle1985continuous}
A.S. Kyle.
\newblock Continuous auctions and insider trading.
\newblock {\em Econometrica: Journal of the Econometric Society}, pages
  1315--1335, 1985.

\bibitem[Las04]{lasserre.2004}
G.~Lasserre.
\newblock Asymmetric information and imperfect competition in a continuous time
  multivariate security model.
\newblock {\em Finance and Stochastics}, 8(2):285--309, 2004.

\bibitem[MS82]{milgrom1982information}
P.~Milgrom and N.~Stokey.
\newblock Information, trade and common knowledge.
\newblock {\em Journal of Economic Theory}, 26(1):17--27, 1982.

\bibitem[O'H07]{o2007market}
M.~O'Hara.
\newblock {\em Market microstructure theory}.
\newblock Blackwell, 2007.

\bibitem[Pro04]{protter2004stochastic}
P.E. Protter.
\newblock {\em Stochastic integration and differential equations, second
  edition}.
\newblock Springer Verlag, 2004.

\bibitem[Ver10]{veraart2010optimal}
L.A.M. Veraart.
\newblock Optimal market making in the foreign exchange market.
\newblock {\em Applied Mathematical Finance}, 17(4):359--372, 2010.

\bibitem[Zen03]{zeng2003partially}
Y.~Zeng.
\newblock A partially observed model for micromovement of asset prices with
  {B}ayes estimation via filtering.
\newblock {\em Mathematical Finance}, 13(3):411--444, 2003.

\end{thebibliography}
\bibliographystyle{alpha}

\end{document}